\documentclass[12pt]{article}

\usepackage{hyperref}

\usepackage{amsfonts}
\usepackage{amsmath}
\usepackage{amssymb}
\usepackage{amscd}
\usepackage{amsthm}
\usepackage{indentfirst}

\newtheorem{thm}{Theorem}[section]
\newtheorem{lemma}[thm]{Lemma}

\newtheorem{defin}[thm]{Definition}
\newtheorem{rem}[thm]{Remark}
\newtheorem{exam}[thm]{Example}

\newtheorem{question}[thm]{Question}

\newcommand{\R}{{\mathbb{R}}}

\newcommand{\N}{{\mathbb{N}}}
\newcommand{\C}{{\mathbb{C}}}

\newcommand{\cB}{{\mathcal{B}}}

\newcommand{\cL}{{\mathcal{L}}}

\newcommand{\cS}{{\mathcal{S}}}

\def\id{{1\hskip-2.5pt{\rm l}}}

\newcommand{\cC}{{\mathcal{C
}}}

\newcommand{\tr}{{\rm tr}}

\DeclareMathOperator{\sgrad}{sgrad}

\newcommand{\End}{\textup{End}}
\newcommand{\Vol}{\textup{Vol}}
\newcommand{\Int}{\displaystyle{\int}}

\newcommand{\om}{\omega}

\newcommand{\Ga}{\Gamma}

\newcommand{\bigo}{\mathcal{O}}
\newcommand{\op}{{\textrm{op}}}
\usepackage{color}

\begin{document}
\title{Berezin-Toeplitz quantization and the least unsharpness principle}
\author{Louis Ioos$^1$,
David Kazhdan and Leonid Polterovich}
\maketitle
\newcommand{\Addresses}{{
  \bigskip
  \footnotesize

  \textsc{School of Mathematical Sciences, Tel Aviv University, Ramat Aviv, Tel Aviv 69978
Israel}\par\nopagebreak
  \textit{E-mail address}: \texttt{louisioos@mail.tau.ac.il}
	
	\medskip
	
	\textsc{Einstein Institute of Mathematics, Hebrew University, Givat Ram, Jerusalem, 91904, Israel}\par\nopagebreak
  \textit{E-mail address}: \texttt{kazhdan@math.huji.ac.il}
	\medskip
	
	  \textsc{School of Mathematical Sciences, Tel Aviv University, Ramat Aviv, Tel Aviv 69978
Israel}\par\nopagebreak
  \textit{E-mail address}: \texttt{polterov@tauex.tau.ac.il}

}}

\footnotetext[1]{Partially supported by the European Research Council Starting grant 757585}

\begin{abstract}
We show that compatible almost-complex structures on symplectic manifolds  correspond to optimal quantizations.
\end{abstract}

\tableofcontents

\section{Introduction}
A Riemannian metric on a symplectic manifold $(M,\omega)$ is
$\omega$-compatible if it can be written as
\begin{equation}
\label{eq-metricg}
g_{\omega,J}(\cdot,\cdot):= \omega(\cdot,J\cdot)\;,
\end{equation}
where $J$ is an almost complex structure on $M$. Vice versa, an almost complex structure
$J$ is $\omega$-compatible if the bilinear form \eqref{eq-metricg}
is a Riemannian metric. Compatible geometric structures were introduced as an effective auxiliary tool for detecting rigidity phenomena on symplectic manifolds \cite{Gromov}. In the present paper we show that these structures naturally arise from the perspective of mathematical physics. Loosely speaking, they correspond to ``optimal" quantizations, the ones minimizing a natural physical quantity called \emph{unsharpness},
which is one of the main characters of this paper (see Section
\ref{sec-coc} below).

Quantization is a mathematical recipe behind the quantum-classical correspondence, a fundamental physical principle stating that quantum mechanics contains classical mechanics in the limiting regime when the Planck constant $\hbar$ tends to zero \cite{D}. There exist two different, albeit related mathematical models of this principle. Assume that the classical phase
space is represented by a closed (i.e., compact without boundary) symplectic manifold $(M,\omega)$. The first model,  {\it deformation quantization}, is a formal associative deformation
$$f*g= fg+\hbar c_1(f,g) + \hbar^2 c_2(f,g) + \cdots$$
of the multiplication on the space $C^\infty(M)$ of
smooth real functions on $M$ such that
$f*g-g*f= i\hbar\{f,g\}+\bigo(\hbar^2)$,  where $\{\cdot,\cdot\}$ stands for the Poisson bracket \cite{BFFLS}. The operation $*$ is called {\it the star-product} and the Planck constant $\hbar$ plays the role of a formal deformation parameter.

The second model, {\it geometric quantization}, is described as a linear correspondence $f \mapsto T_\hbar(f)$
between classical observables, i.e., real functions
$f$ on the phase space $M$, and quantum observables, i.e., Hermitian operators on a complex Hilbert space. This correspondence is assumed to respect, in the leading order as $\hbar\to 0$, a number of basic operations. In the present paper, we focus on {\it Berezin-Toeplitz quantizations} \cite{Berezin,BMS,Gu,BU,MM,Sc,C},
whose distinctive feature is to send non-negative functions to non-negative operators (see Section \ref{sec-bt}). The known models of Berezin-Toeplitz
quantization on closed symplectic manifolds
(see the discussion following Theorem \ref{BQ}) determine a deformation quantization \cite{BMS, Sch, Gu}, and are provided by certain auxiliary data involving in particular
an almost complex structure $J$  compatible  with the symplectic form on the phase space.
While deformation quantizations of closed symplectic manifolds are
completely classified up to a natural equivalence, the
classification
of Berezin-Toeplitz quantizations is not yet completely
understood (see however \cite{KS01} for the relation between the two).

The main finding of the present paper is that conversely,
any Berezin-Toeplitz quantization, defined through natural axioms
presented in Section \ref{sec-bt},
gives rise in a canonical way to a Riemannian metric
on the phase space. Specifically, we make the natural assumption
that there exists of a
complex-valued bi-differential operator
$c: C^\infty(M) \times C^\infty(M) \to C^\infty(M,\C)$ such that
the $\C$-linear extension of $T_\hbar$ satisfies
\begin{equation}\label{eq-expan}
T_\hbar(f)T_\hbar(g) = T_\hbar(fg)+ \hbar\,T_\hbar(c(f,g))+ \bigo(\hbar^2)\;,
\end{equation}
for all $f,\,g\in C^\infty(M)$ as $\hbar\to 0$, and
we show that this induces a Riemannian metric $G$ on $M$
by the formula
$$c(f,g) = -\frac{1}{2}G(\sgrad f, \sgrad g)\,,$$
where $\sgrad f$ stands for the Hamiltonian vector field of a function $f$ on $M$.
(see Theorem \ref{thm-main-new-1}\,(I) below).
%

Leaving precise definitions for Section \ref{sec-main}, let us discuss the above-mentioned results informally and present a motivation coming from physics. To this end recall that it is classically known, starting from the Groenewold-van Hove theorem, that a Berezin-Toeplitz correspondence cannot be a genuine morphism between the Lie algebras of functions and operators. We focus on yet another constraint on the precision of Berezin-Toeplitz quantizations, which we call {\it unsharpness}, and which is governed by the Riemannian metric $G$ defined above. The notion of unsharpness is closely related to the Heisenberg uncertainty principle. It comes from an analogy between quantization and measurement based on the formalism of positive operator valued measures (POVMs), which serves both subjects, and which we briefly recall in Section \ref{sec-bt}. The unsharpness metric is a particular instance of the noise operator \cite{Busch} describing, loosely speaking, the increment of variances in the process of quantization
(see the discussion p.13).

In this language, we propose {\it the least unsharpness principle}, a variational principle
selecting quantizations whose unsharpness metric
has minimal possible volume on phase space. It turns out that the least unsharpness volume equals the symplectic volume
(see Theorem \ref{thm-main-new-1}\,(II) below),
so that a quantization satisfying the least unsharpness
principle determines a compatible almost complex structure $J$
on $(M,\om)$, in the sense that its unsharpness metric
satisfies $G=g_{\om,J}$ as in
\eqref{eq-metricg}. We refer to Section \ref{sec-coc} for basic properties of unsharpness, while existence of the unsharpness metric and the least unsharpness principle are stated in Section \ref{sec-main} and proved in Section \ref{sec-proofs}.

The unsharpness metric is a natural geometric invariant of a Berezin-Toeplitz quantization, and can be seen as a first step towards classification. As a case study, we show in Section \ref{sec-equiv} that for $SU(2)$-equivariant quantizations of the two-dimensional sphere, the
unsharpness metric completely determines the quantization up to
conjugation and up to second order as $\hbar\to 0$. Further comments on classification can be found
in Section \ref{subsec-class}.

Some historical remarks are in order.  A canonical appearance of Riemannian geometry in quantization
was discussed on a number of occasions in physical literature. Anandan and Aharonov \cite{AA} and Ashtekar and Schilling \cite{AS} developed a geometric approach to quantum mechanics based on the Fubini-Study metric on the projective space of pure quantum states. Klauder (see, e.g., \cite{Klauder}) studied a model of a path-integral quantization where the role of a metric was to define a Brownian motion on the phase space. The idea of selecting optimal quantizations as those possessing the least uncertainty goes back to Gerhenstaber \cite{G}. He deals with quantizations which do not necessarily preserve positivity, and his least uncertainty principle implies that unsharpness identically vanishes on some restricted class of observables
(see Section \ref{subsec-ger} for further discussion). Finally, while classification of equivariant quantizations is known in the context of deformation quantization \cite{AL,BBG}, our setting, including
the notion of equivalence, is substantially different.
The case of $SU(2)$-equivariant Berezin-Toeplitz quantizations of the sphere which we settle in Section \ref{sec-equiv} is, to the best of our knowledge, the first one where a complete classification is currently available.

\section{Berezin-Toeplitz quantization}\label{sec-bt}
Given a finite-dimensional complex Hilbert space $H$, we write
$\cL(H)$ for the space of all Hermitian
operators {(representing quantum observables), and
$\cS(H) \subset \cL(H)$ for the subset of all non-negative trace-one Hermitian operators (representing quantum states)}.

\begin{defin}\label{POVMdef}
{\rm
An $\cL(H)$-valued \emph{positive operator valued measure} (POVM)
on a set $M$ equipped with a $\sigma$-algebra $\cC$ is a map $\Ga: \cC \to \cL(H)$ which satisfies the following conditions:
\begin{itemize}
\item $\Ga(\emptyset) = 0~\text{and}~\Ga(M)=\id$\,;
\item $\Ga(X) \geq 0$ for every $X \in \cC$\,;
\item ($\sigma$-additivity) $\Ga\left(\bigsqcup_{i\in\N} X_i\right) = \sum_{i\in\N} \Ga(X_i)$ for any sequence
of pair-wise disjoint subsets $\{X_i \in \cC\}_{i\in\N}$.
    \end{itemize}}
\end{defin}

According to \cite{CDS}, { for
 every $\cL(H)$-valued POVM measure $\Ga$, there exist a probability measure $\alpha$ on $(M, \cC)$ and  a measurable function  $F: M \to \cS(H)$ such that

\begin{equation}\label{eq-POVM-density}
d\Ga = n\,Fd\alpha\;,
\end{equation}
where $n= \dim_{\C} H$.}

\begin{rem}\label{cohstate}
{\rm
In the context of quantization, the state
$F_{x} \in \cS(H)$ is
called the {\it coherent state} associated with $x\in M$,
and describes the quantization of a classical particle
sitting at $x\in M$ in phase space.}
\end{rem}

For any real classical observable $f \in L^1(M,\alpha)$,
we  define the  \emph{quantization} $T(f)$ as the integral
\begin{equation}\label{quantmap}
T(f):= \int_M f\,d\Ga \in \cL(H)\,.
\end{equation}
The dual map $T^*: \cL(H) \to L^\infty(M)$ with respect to the scalar
product $((A,B)):= \tr(AB), A,\,B\in\cL(H)$ satisfies
$T^*(A)(x) = n\,\tr (AF(x))$, for any $x\in M$ and
$A\in\cL(H)$.

\begin{rem}{\rm For a quantum observable $A$, the function
$T^*(A)\in L^\infty(M)$ has a natural interpretation as the classical observable whose value
at $x\in M$ is the expectation value of $A$ at the associated
coherent state $F_x$. Thus, we call $T^*(A)$ the \emph{dequantization} of the quantum
observable $A \in\cL(H)$.}
\end{rem}
\begin{defin}\label{Bdef}
{\rm
The composition
\begin{equation}
\label{eq-BT}
\begin{split}
\cB:= \frac{1}{n}T^*T: L^1(M,\alpha) &\longrightarrow L^\infty(M)\,,\\
f(x)~& \longmapsto ~\tr\left(T(f)F_x\right)
\end{split}
\end{equation}
is called the {\it Berezin transform} associated to the POVM $\Gamma$.}
\end{defin}

\begin{rem}{\rm  The Berezin transform can
be interpreted as quantization followed by dequantization. It is a
measure of the blurring induced by quantization.}
\end{rem}

To study the quantum-classical correspondence, we need to introduce
a parameter $\hbar$ in the above story, which can be thought
as the \emph{Planck constant}, and from which we recover the laws
of classical mechanics as $\hbar\to 0$. This is given a precise
meaning  via the following definition.

\begin{defin}\label{BTdef} {\rm Let $(M,\omega)$ be  a closed connected symplectic manifold of dimension $2d$ and $\cC$ be the $\sigma$-algebra of its Borel sets in $M$.}
 {\rm
A \emph{Berezin-Toeplitz
quantization} of $M$ is the following data:
\begin{itemize}
\item a subset $\Lambda \subset \R_{>0}$ having $0$ as limit
point ;
\item a finite-dimensional complex Hilbert space $H_\hbar$
for each $\hbar \in \Lambda$ ;
\item an $\cL(H_\hbar)$-valued positive operator valued measures $\Ga_\hbar$ on $M$
for each $\hbar\in\Lambda$,
\end{itemize}
such that the \emph{Toeplitz map}
$T_\hbar: C^\infty(M,\C) \to \End(H_\hbar)$ induced for all
$\hbar\in\Lambda$ by the quantization
map \eqref{quantmap}
is surjective and satisfies the following
estimates, {\it uniformly} in the $C^N$-norms of
$f,\,g\in C^\infty(M)$ for some $N\in\N$ :
\begin{itemize}
\item[{(P1)}] {\bf (norm correspondence)} $$\|f\|- \bigo(\hbar) \leq \|T_\hbar(f)\|_{\op} \leq \|f\|\;,$$
    where $\|\cdot\|_{op}$ is the operator norm and $\|f\|:=
    \max_{x\in M} |f(x)|$ ;
\item[{(P2)}] {\bf (bracket correspondence)} $$\left\| -\frac{i}{\hbar} [T_{\hbar}(f),T_{\hbar}(g)] - T_\hbar (\{f,g\})\right\|_{\op} =\bigo(\hbar)\;,$$
    where $[\cdot,\cdot]$ stands for the commutator and
    $\{\cdot,\cdot\}$ for the Poisson bracket\footnote{Our convention for the Poisson bracket is
    $\{f,g\}:=-\om(\text{sgrad} f,\text{sgrad} g)$ for all
    $f,\,g\in C^\infty(M)$, where
    $\text{sgrad} f$ is the Hamiltonian vector field of
    $f$ defined by $\iota_{\text{sgrad} f}\om+df=0$.};
\item[{(P3)}] {\bf (quasi-multiplicativity)} There exists a bi-differential operator\\
$c: C^\infty(M) \times C^\infty(M) \to C^\infty(M,\C)$
such that
$$\|T_\hbar(f)T_\hbar(g) - T_\hbar(fg)-\hbar T_\hbar(c(f,g))\|_{op}= \bigo(\hbar^2)\;.$$
\item[{(P4)}] {\bf (trace correspondence)}  $$  \tr(T_\hbar(f)) = (2\pi\hbar)^{-d}\int_M f\,R_\hbar\, d\mu\;,$$
 where $R_\hbar\in C^\infty(M)$ satisfies $R_\hbar = 1+ \bigo(\hbar)$, and $d\mu = \frac{\omega^d}{d!}$ is the symplectic volume on $M$;
 \item[{(P5)}] {\bf (reversibility)} The maps
 $\cB_\hbar:C^\infty(M)\to C^\infty(M)$ induced by the
 Berezin transform \eqref{eq-BT} satisfy $$\cB_\hbar (f) = f + \bigo(\hbar)\,.$$
\end{itemize}}
\end{defin}

By uniformly in the $C^N$-norms of $f,\,g\in C^\infty(M)$,
we mean that there exists a constant $C>0$ such that,
in axioms (P2) and (P3) and for $k=1,\,2$ respectively, the
remainders satisfy
\begin{equation}
\left|\bigo(\hbar^k)\right|\leq C\hbar^k\,\|f\|_{C^N}\|g\|_{C^N}\,,
\end{equation}
while in axioms (P1) and (P5), the remainders satisfy
$\left|\bigo(\hbar)\right|\leq C\hbar\,\|f\|_{C^N}$.
Writing the density
\eqref{eq-POVM-density} associated to $\Ga_\hbar$ in the form
\begin{equation}\label{eq-POVM-density-1}
d\Ga_{\hbar}(x) = n_\hbar\,F_{\hbar,x}\,d\alpha_\hbar(x)\;,
\end{equation}
the trace correspondence (P4) implies
$$n_\hbar= \frac{\text{Vol}(M,\omega) + \bigo(\hbar)}{(2\pi \hbar)^d}\;,$$
and
\begin{equation}\label{eq-measure}
d\alpha_\hbar = \frac{1+\bigo(\hbar)}
{\text{Vol}(M,\omega)}d\mu\;.
\end{equation}
where $\text{Vol}(M,\om)>0$ denotes the symplectic volume of
$(M,\omega)$.

\bigskip

The existence of a Berezin-Toeplitz quantization is a highly non-trivial result. To discuss it,
recall that an almost complex structure $J$ on $M$ is {$\omega$-{\it compatible} if the form
$G_J:=\omega(\cdot,J\,\cdot)$ is a Riemannian metric on $M$. We refer to $(M,\omega, J, G_J)$ as an \emph{almost-K\"{a}hler structure} on $M$.

\begin{thm}\label{BQ}
If the closed symplectic manifold $(M,\omega)$
is \emph{quantizable}, i.e., if the cohomology class
$[\omega]/(2\pi)$ is integral, then every $\omega$-compatible
almost-complex structure $J$ defines a Berezin-Toeplitz
quantization of $(M,\om)$, with $\Lambda=\{1/k\}_{k\in\N}$.
\end{thm}

In the case of \emph{K\"{a}hler manifolds}, i.e., if we assume
additionally that the almost complex structure $J$ is integrable,
there is a canonical
construction of a Berezin-Toeplitz quantization,
where the Hilbert spaces
$H_\hbar$ consist of the global \emph{holomorphic sections} of
a
holomorphic Hermitian line bundle with Chern curvature equal to
$-2\pi i k\om$, with $\hbar=1/k$, and the associated Toeplitz map
$T_\hbar$ sends $f\in C^\infty(M)$ to the multiplication by $f$
followed by the orthogonal $L_2$-projection on holomorphic sections.
In this context, Theorem \ref{BQ} has been established by Bordemann,
Meinrenken and Schlichenmaier in \cite{BMS}, using the theory of
Toeplitz structures developed by Boutet de Monvel
and Guillemin in \cite{BdMG81}. The fact that this theory extends
to the almost-K\"{a}hler case was proved in a series of papers by Guillemin \cite{Gu}, Borthwick and Uribe \cite{BU}, Schiffman and Zelditch \cite{SZ02}, Ma and Marinescu \cite{MM}, Charles \cite{C} and the first named author, Lu, Ma and Marinescu
\cite{I2}.
The dependence of the remainders in terms of the derivatives of
the functions is discussed in \cite{CP}.

\section{Unsharpness cocycle} \label{sec-coc}

In this section, we study general properties of the bi-differential
operator $c: C^\infty(M) \times C^\infty(M) \to C^\infty(M)$
from the quasi-multiplicativity property (P3) of a
Berezin-Toeplitz quantization.
First note that norm correspondence
(P1) implies that, if an asymptotic expansion such
as the one appearing in (P3) holds, then it is unique,
and in particular, the bi-differential operator
$c$ is uniquely defined. Then
the associativity of the
composition of operators implies that $c$ is a
\emph{Hochschild cocycle}, meaning that for all
$f_1,\,f_2,\,f_3 \in C^\infty(M)$, we have
\begin{equation}\label{Hochcocycle}
f_1\,c(f_2,f_3)-c(f_1f_2,f_3)+c(f_1,f_2f_3)-c(f_1,f_2)\,f_3=0\;.
\end{equation}
Denote by $c_-$ and $c_+$ its anti-symmetric and symmetric parts, respectively:
$$c_-(f,g) := \frac{c(f,g)-c(g,f)}{2}\quad\text{and}
\quad c_+(f,g) := \frac{c(f,g)+c(g,f)}{2}\;.$$
By the bracket correspondence (P2), we {have
$T_\hbar(2c_-(f,g) - i\{f,g\}) = \bigo(\hbar)$},
and hence by the norm correspondence (P1), we get the formula
\begin{equation}\label{eq-cminus}
c_-(f,g)= \frac{i}{2}\{f,g\}\;.
\end{equation}
Thus the anti-symmetric part $c_-$ (responsible for the  non-commutativity of quantum observables) does not depend  on a choice of a quantization. In contrast , the symmetric part $c_+$
does depend on a choice of a quantization. By the quasi-multiplicativity
property (P3), the cocycle $c_+$ associated to a Berezin-Toeplitz quantization
measures its failure of being a multiplicative morphism on Poisson-commutative subspaces of $C^\infty(M)$.
From formula \eqref{eq-cminus} and basic properties
of the Poisson bracket, we know that $c^-$, hence also $c^+$,
satisfy formula \eqref{Hochcocycle} for a Hochschild cocycle.

\begin{defin}\label{def-unsharp} {\rm We say that
$c_+$ is  the \emph{unsharpness cocycle}
of a quantization or simply its  \emph{unsharpness}.}
\end{defin}}

Note that by formula \eqref{quantmap}, the operator
$T_\hbar(f)\in\End(H_\hbar)$ is Hermitian
if and only if $f\in C^\infty(M,\C)$ is real valued,
and as the square of a Hermitian operator is Hermitian,
the quasi-multiplicativity property (P3) then shows that
$c^+:C^\infty(M)\times C^\infty(M)\to C^\infty(M)$
is a real-valued bi-differential operator.
It is also a symmetric Hochschild cocycle, so that as explained e.g. in
\cite[Prop.\,2.14]{RG}, it is a differential Hochschild
\emph{coboundary}. This means that there exists
a real-valued differential operator
$a:C^\infty(M)\to C^\infty(M)$
such that
\begin{equation}\label{eq-a}
c_+(f,g) = a(fg)-f\,a(g)-g\,a(f)\;.
\end{equation}
for all $f,\,g\in C^\infty(M)$. Since $T_\hbar(1)=\id$, we have that $c_+(1,1)=0$,
and therefore $a(1)=0$. Note that $a$ is determined
up to its first order part.
The following result shows
that the positivity preserving property imposes a strong condition on $c_+$.

\begin{thm}\label{thm-order} The bi-differential operator $c_+$ is of order $(1,1)$.
\end{thm}
The proof is given in Section \ref{sec-order-proof} below.
Theorem \ref{thm-order} sheds light on the differential
operator $a$ appearing in the coboundary formula \eqref{eq-a}.
In fact, let us choose some
Darboux coordinates $U\subset X$, and take
$f,\,g\in C^\infty(M)$ with compact support in $U$.
In these coordinates, we can write
\begin{equation}\label{c+coord}
c_+(f,g)=\sum_{j,\,k=1}^{2d} a_{jk}\,\partial_j f\,
\partial_k g\,,
\end{equation}
with smooth $a_{jk}=a_{kj}\in C^\infty(U)$
for each $1\leq j,\,k\leq 2d$.
Then one can choose the differential operator
\begin{equation}\label{acoord}
a:=\frac{1}{2}\sum_{j,\,k=1}^{2d}\partial_j \left(a_{jk}\partial_k\right)\,.
\end{equation}
in the coboundary formula \eqref{eq-a}.
Using integration by parts,
we see that \eqref{acoord} is symmetric with respect to the
canonical $L^2$-scalar product on $C^\infty(M)$ associated
to the symplectic volume, and as the differential
operator $a$ is determined up to
its first order part, it is the unique such choice.

\begin{exam}\label{exam-1}{\rm Assume that $(M,\om)$ is
quantizable and equipped with an almost-K\"{a}hler structure
$(M,\om,J,G_J)$. Then the induced Berezin-Toeplitz quantization of
Theorem \ref{BQ} satisfies
\begin{equation}
\label{eq-cpluskahler1}
c_+(f,g) = -\frac{1}{2} (\nabla f,\nabla g)\;,
\end{equation}
where the gradient and the product are defined with respect to $G_J$.
Using that
\begin{equation}\label{Deltacocycle}
\Delta(fg) + 2(\nabla f,\nabla g) = f\Delta g + g\Delta f\;,
\end{equation}
where $\Delta$ is the (positive) Laplace-Beltrami operator associated
with $G_J$, the differential operator in the
coboundary formula \eqref{eq-a}
can then be chosen to be $a=\Delta/4$, and
by basic properties of $\Delta$, it is the unique
$L^2$-symmetric choice with respect to the symplectic volume
form, as it coincides with the Riemannian volume form of $G_J$.
Formula \eqref{eq-cpluskahler1} can be found in
\cite[p.\,257]{Xu} for the K\"{a}hler case and in
\cite{I1,I2} for the almost-K\"{a}hler case. Using the $J$-invariance
of the metric and the relation $J\,\text{sgrad} f= - \nabla f$
between Hamiltonian vector field and gradient of a function
$f\in C^\infty(M)$
for an $\omega$-compatible metric, formula \eqref{eq-cpluskahler1}
translates into
\begin{equation}
\label{eq-cpluskahler}
c_+(f,g) = -\frac{1}{2} G_J(\text{sgrad} f,\text{sgrad}g)\;.
\end{equation}
}
\end{exam}

\begin{exam}\label{exam-2} {\rm   We  now give an example of a
Berezin-Toeplitz quantization whose unsharpness cocycle
$c_+$ is not of the form \eqref{eq-cpluskahler} for some
almost-K\"{a}hler structure on $(M,\om)$. This example serves as a paradigm
for the construction presented in the proof of one of our main results,
Theorem \ref{thm-main-new-1}(III) below. Assume $(M,\om)$ quantizable and equipped
with an almost-K\"{a}hler structure $(M,\om,J,G_J)$, and
consider the induced Berezin-Toeplitz quantization of
Theorem \ref{BQ}. Fix $t>0$, and using the notations of
Example \ref{exam-2}, consider for any
$\hbar\in\Lambda=\{1/k\}_{k\in\N}$ the map
$T_\hbar^{(t)}: C^\infty(M) \to \cL(H_\hbar)$
defined for any $f\in C^\infty(M)$ by
$$T^{(t)}_\hbar(f) := T_\hbar(e^{-t\hbar\Delta} f)\;.$$
Observe that the heat flow preserves positivity, so that
$T^{(t)}_\hbar$ is in fact the quantization map
\eqref{quantmap} induced by a POVM construction.
Then from the classical small time asymptotic expansion
of the heat kernel (see e.g. \cite[Th.\,2.29, (2.8)]{BGV04}),
as $\hbar\to 0$, we have
\begin{equation}
e^{-t\hbar\Delta}f = f -t\hbar\Delta\,f+\bigo(\hbar^2)\,
\|f\|_{C^4}
\,,
\end{equation}
and this implies in particular that all the axioms of
Definition \ref{BTdef} hold. Let us now calculate the associated
unsharpness cocycle, denoted by
$c_+^{(t)}$. For any $\hbar\in\Lambda$ and $A,\,B\in\End(H_\hbar)$,
put $A \bullet B := \frac{1}{2}(AB+BA)$, and recall formula
\eqref{Deltacocycle} for the Laplace-Beltrami operator.
Then as $\hbar\to 0$, we have
\begin{equation}\label{computexam1}
\begin{split}
&T^{(t)}_\hbar(f)\bullet T^{(t)}_\hbar(g)=T_\hbar(f) \bullet T_\hbar(g) -
t \hbar T_\hbar(f\Delta g + g\Delta f)+ \bigo(\hbar^2)\\
&= T_\hbar(fg) - \hbar T_\hbar\left(\frac{1}{2}(\nabla f,\nabla g) +
t\left(\Delta(fg) + 2(\nabla f,\nabla g)\right)\right) + \bigo(\hbar^2)\\
&= T_\hbar(fg) + \hbar
T^{(t)}_\hbar\left(-\left(\frac{1}{2}+2t\right) (\nabla f,\nabla g)\right) + \bigo(\hbar^2)\;,
\end{split}
\end{equation}
so that, recalling that the quasi-multiplicativity property
(P3) determines the unsharpness
cocycle uniquely via norm correspondence (P1), we get
\begin{equation}
\begin{split}
c^{(t)}_+(f,g)&=-\left(\frac{1}{2}+2t\right) (\nabla f,\nabla g)\\
&=-\frac{1}{2}\left(1+4t\right)G_J(\text{sgrad} f,\text{sgrad} g)\;.
\end{split}
\end{equation}
In particular, we see that $c^{(t)}_+$ is of the form
\eqref{eq-cpluskahler} for the Riemannian metric
$G^{(t)}:=(1+4t)\,G_J$ on $M$, whose volume
is strictly bigger than the volume of the almost-K\"{a}hler metric
$G_J$. As the volume of an almost-K\"{a}hler metric is always
equal to the symplectic volume of $(M,\om)$, we see
from \eqref{exam-1} that
$c^{(t)}_+$ is not the unsharpness
cocycle of a Berezin-Toeplitz quantization coming from
Theorem \ref{BQ}.
}
\end{exam}

\section{The least unsharpness principle}
\label{sec-main}

\bigskip

In this section, we state the main theorem on unsharpness
of Berezin-Toeplitz quantizations, which we call the
\emph{least unsharpness principle}, and discuss its physical meaning.

Recall from Theorem \ref{thm-order} that
the unsharpness cocycle $c_+$ of a Berezin-Toeplitz
quantization is a bi-differential operator of order $(1,1)$,
so that
there exists a bilinear symmetric form $G$ on $TM$ such that
\begin{equation}\label{eq-G}
c_+(f,g) =: -\frac{1}{2}G(\text{sgrad}f,\text{sgrad}g)\,,
\end{equation}
where $\text{sgrad}f,\,\text{sgrad}g$ denote
the Hamiltonian vector fields of  $f,\,g\in C^\infty(M,\R)$.
Our main result provides a description
of this bilinear form $G$.

\begin{thm}\label{thm-main-new-1} Let $(M,\omega)$ be a closed symplectic manifold.
\begin{itemize}
\item[\textup{(I)}] For every Berezin-Toeplitz quantization of $M$, the form $G$ is a Riemannian metric on $M$ which can be written
as { the sum
\begin{equation}\label{metric}
G=\om(\cdot,J\cdot)+\rho(\cdot,\cdot)\,,
\end{equation}
where $J\in\End(TM)$ is a compatible almost complex structure on
$(M,\om)$ and $\rho$ is a non-negative symmetric bilinear form
on $TM$.
\item[\textup{(II)}] We have $\textup{Vol}(M,G) \geq \textup{Vol}(M,\omega)$, with
equality if and only if $\rho \equiv 0$.}
\item[\textup{(III)}] If $(M,\omega)$ is quantizable, then every Riemannian metric of the form \eqref{metric} arises from some Berezin-Toeplitz quantization.
    \end{itemize}
    \end{thm}

 \medskip
 \noindent The proof is given in Section \ref{sec-proofs}. Let us mention that
the proof of item (III) of the theorem is modeled on Example \ref{exam-2}  above and is constructive. We produce the desired Berezin-Toeplitz quantization with the unsharpness metric given by \eqref{metric} as the composition of the almost-K\"{a}hler quantization associated to $(\omega,J)$ and an explicit, albeit non-canonical, Markov operator depending on all the data including $\rho$.

    \medskip
    \noindent
    \begin{rem}\label{rem-rho}{\rm For a given metric $G$ on $M$,
    the decomposition \eqref{metric}
    is in general not unique. However, as the proof of Theorem \ref{thm-main-new-1}\,(I)
    will show,  { there exists a unique $\omega$-compatible almost complex structure $J$
     which additionally is $G$-orthogonal, i.e., $G(J\xi,J\eta)=G(\xi,\eta)$ for all $\xi,\eta\in TM$. Furthermore, for such
    a $G$, the symmetric bilinear form $\rho(\xi,\eta) = G(\xi,\eta) - \omega(\xi,J\eta)$ is
     non-negative, thus providing decomposition \eqref{metric}.
%
    }}
    \end{rem}

Before going to the proof of Theorem \ref{thm-main-new-1} in
the next Section, let us first discuss the physical meaning of
the unsharpness cocyle $c^+$ associated with a Berezin-Toeplitz
operator, which shows from general principles
that it is at least non-negative.
{With every quantum state $\theta \in \cS(H_\hbar)$ one associates a classical state (called {\it the Husimi measure}), which is the  probability measure $\mu_{\theta}$ on $M$ such that
\begin{equation}\label{Husimi}
\int_M f\,d\mu_\theta = \tr(T_\hbar(f)\,\theta)\;, f \in C^\infty(M)\,.
\end{equation}
}
This equality can be interpreted as follows: the expectation of any classical observable
$f$ in the classical state $\mu_\theta$ coincides with the expectation of the corresponding quantum observable $T_\hbar(f)$ in the state $\theta$. What happens with variances? It turns out
that the quantum variance is in general bigger than the classical one. More precisely, we
have that
$$\mathbb{V}ar (f,\mu_\theta) = \int_M f^2 d\mu_\theta - \left(\int fd\mu_\theta\right)^2\;,$$
$$\mathbb{V}ar (T_\hbar(f),\theta) = \tr (T_\hbar(f)^2\theta) - \left(\tr(T_\hbar(f)\theta)\right)^2\;,$$
and hence
\begin{equation}\label{noiseincrement}
\mathbb{V}ar (f,\mu_\theta) = \mathbb{V}ar (T_\hbar(f),\theta) + \tr(\Delta_\hbar(f)\theta)\;,
\end{equation}
where
\begin{equation}\label{noisedef}
\Delta_\hbar(f) := T_\hbar(f^2)-T_\hbar(f)^2\;.
\end{equation}
The operator $\Delta_\hbar(f)$ is called the {\it noise operator} (see e.g. \cite{Busch}), whose main property is
that it is a \emph{non-negative operator}.
It describes the increase of variances, which
can be interpreted as the unsharpness of the quantization.
Then by the quasi-multiplicativity property (P3), we have
\begin{equation}\label{eq-noise}
\Delta_\hbar(f) = - \hbar\,T_\hbar\left(c_+(f,f)\right)+ \bigo(\hbar^2)\,.
\end{equation}
Look at the expectation of $\Delta_\hbar(f)$ at the coherent state $F_{\hbar,x}$ of Remark \ref{cohstate} associated to
$\Gamma_\hbar$,
\begin{equation}\label{noisefla}
\begin{split}
\tr&\left(\Delta_\hbar(f)F_{\hbar,x}\right)= -\hbar\,\tr \left(T_\hbar\,(c_+(f,f))F_{\hbar,x}\right)+\bigo(\hbar^2)\\
&=-\hbar\,\cB_\hbar\left(c_+(f,f)\right)(x)+\bigo(\hbar^2)\\
&=-\hbar\,c_+(f,f)(x)+\bigo(\hbar^2)\;.
\end{split}
\end{equation}
This explains the name of unsharpness cocycle for $c^+$.
Since the noise operator is non-negative, we get the following
fundamental property of the unsharpness cocycle,
\begin{equation}\label{eq-positivity}
- c_+(f,f)(x) \geq 0\quad\text{for all}\quad x \in M\;.
\end{equation}
This shows that the symmetric
bilinear form $G$ defined in equation
\eqref{eq-G} is at least semi-positive.
This property is the first step towards the proof of
Theorems \ref{thm-main-new-1}, showing that $G$ is in fact
a Riemannian metric, called the \emph{unsharpness metric} of
the quantization. Note that this property is also at the basis
of the proof of Theorem \ref{thm-order}.

\medskip

Define the {\it total unsharpness} of a Berezin-Toeplitz quantization
as the volume of the phase space $M$ with respect to the unsharpness metric. With this language, statement (II) of Theorem \ref{thm-main-new-1} can be interpreted as the {\it least unsharpness principle}: the minimal possible total unsharpness equals the symplectic volume, and the least unsharpness metrics
are induced by compatible almost-complex structures on $M$.

\begin{rem}\label{rmkvar}
{\rm Let us assume that the Berezin transform
admits an asymptotic expansion
up to the first order as $\hbar\to 0$ of the following form
for all $f \in C^\infty(M)$,
\begin{equation}\label{BTexp}
\cB_\hbar(f) = f + \hbar Df + \bigo(\hbar^2)\,,
\end{equation}
where $D$ is a differential operator, stenghtening the
reversibility property (P5).
Then by Definition \ref{Bdef} of the Berezin transform,
formula \eqref{quantmap} for the quantization map
and the expansion \eqref{eq-measure} for $\alpha_\hbar$,
for all $f,\,g \in C^\infty(M)$, we have
\begin{equation}
\begin{split}
\frac{1}{n_\hbar} \tr (T_\hbar(f)T_\hbar(g))&= \int_M \cB_\hbar(f)\,g
\,d\alpha_\hbar\\
&=
\int_M fg\,d\alpha_\hbar + \hbar \int_M \left(Df\right) g\,d\mu + \bigo(\hbar^2)\;.
\end{split}
\end{equation}
On the other hand, by the quasi-mutliplicativity property (P3),
using formula \eqref{eq-cminus} and basic properties
of the Poisson bracket, we get
\begin{equation}
\begin{split}
&\frac{1}{n_\hbar} \tr (T_\hbar(f)T_\hbar(g))= \int_M (fg+\hbar c(f,g)+ \bigo(\hbar^2))\,d\alpha_\hbar\\
& = \int_M fg\,d\alpha_\hbar + \hbar \int_M c_+(f,g)\, d\mu+\bigo(\hbar^2)\;.
\end{split}
\end{equation}
Then taking
$f,\,g\in C^\infty(M)$ with compact support in Darboux
coordinates, using formulas \eqref{c+coord} and
\eqref{acoord} and intergration by parts,
we then get
\begin{equation}
\label{eq-bda}
D=-2a\,,
\end{equation}
where $a:C^\infty(M)\to C^\infty(M)$ is the unique
$L^2$-symmetric differential
operator on $C^\infty(M)$ with respect to symplectic volume
satisfying the coboundary formula \eqref{eq-a}.
In light of Example \ref{exam-1}, this fact generalizes the
Karabegov-Schlichenmaier expansion
\cite{KS01,IKPS} for the Berezin-Toeplitz quantizations of
Theorem \ref{BQ}.

Another consequence of the improvement \eqref{BTexp} of the
reversibility property (P5) is that "unsharpness equals
variance on coherent states". To see this, recall
definition
\eqref{Husimi} of the Husimi measure on the coherent
state $F_{\hbar,x}\in\cS(H_\hbar)$ of Remark
\ref{cohstate}. Then the discussion above implies
\begin{equation}
\begin{split}
\mathbb{V}ar (f,\mu_{F_{x,\hbar}})
&=\cB_\hbar(f^2)-\cB_\hbar(f)^2\\
&=-2\hbar\,c^+(f,f)+\bigo(\hbar^2)\,.
\end{split}
\end{equation}
Thus by formula \eqref{noiseincrement} and \eqref{noisefla},
we get
\begin{equation}\label{var=c+}
\mathbb{V}ar (T_\hbar(f),F_{\hbar,x}) = - \hbar\,c_+(f,f)(x) + \bigo(\hbar^2)\;,
\end{equation}
so that the variance of a quantized observable at
coherent states is equal to its unsharpness.

In their geometric formulation of quantum mechanics,
Ashtekar and Shilling \cite[\S\,3.2.3,\,(26)]{AS}
consider the projectivization $\mathbb{P}(H_\hbar)$
as a ``quantum phase space": a line $\xi\in \mathbb{P}(H_\hbar)$ is identified
with the pure state given by the rank-one projector to $\xi$.
In this setting, they
give a physical interpretation
of the Fubini-Study metric $g^{FS}$ on
$\mathbb{P}(H_\hbar)$ in terms of the variance
of a quantum observable $A\in\cL(H_\hbar)$ at a pure state $\xi \in\mathbb{P}(H_\hbar)$. Specifically,
write $v_A$ for the vector field on $\mathbb{P}(H_\hbar)$
induced by the infinitesimal action of
$iA\in\mathfrak{u}(H_\hbar)$, seen
as an element of the Lie algebra of the group of unitary
operators $U(H_\hbar)$ acting on $\mathbb{P}(H_\hbar)$.
Then
the variance of $A$ at $\xi$ is given by
\begin{equation}
\mathbb{V}ar (A,\xi) = \frac{1}{2}\,
g^{FS}_{\xi}(v_A,v_A)\;.
\end{equation}
Back to the quantization, assume further that the coherent states
$F_{\hbar,x}\in\cS(H_\hbar)$ are pure
for all $x\in M$. Consider the induced map
\begin{equation}\label{Kodemb}
F_{\hbar}:M\longrightarrow\mathbb{P}(H_\hbar)\;.
\end{equation}
Then equation \eqref{var=c+} says that the Fubini-Study length of the
vector field $v_A$ induced by the quantum observable
$A:= T_\hbar(f)\in\cL(H_\hbar)$ at the coherent state
$F_{\hbar,x}\in\mathbb{P}(H_\hbar)$
approaches, as $\hbar\to 0$, the length of the Hamiltonian vector field $\sgrad f$ at a point $x\in M$
with respect to our unsharpness metric.
In the case of the K\"{a}hler quantizations of Theorem \ref{BQ},
the map \eqref{Kodemb} coincides with the
\emph{Kodaira map}. Then the picture described
above is closely related to
a theorem of Tian \cite{Ti} showing that the pullback of the
Fubini-Study metric by the Kodaira map approaches
the Kähler metric as $k\to+\infty$.
}
\end{rem}

\section{Proof of the main Theorem} \label{sec-proofs}

In this Section, we prove Theorem \ref{thm-main-new-1}.
To this end, first recall from the previous section that
the non-negativity of the noise operator \eqref{noisedef}
leads to the semi-positivity property \eqref{eq-positivity}
for the unsharpness. To establish the stronger statement
(i) of Theorem \ref{thm-main-new-1}, we will use a stronger
property of noise operators coming from the general
theory of POVM-based quantum measurements, called the
\emph{noise inequality}.
It appears in several sources
\cite{Ozawa}, \cite[Theorem 7.5]{Hyashi}, \cite[Theorem 9.4.16]{PR}, albeit none of them contains the version we need. For the sake of completeness, we present a proof which closely follows \cite{PR} and is based on an idea from \cite{Janssens}.

Consider a set $M$ equipped with a $\sigma$-algebra $\cC$
together with a finite-dimensional Hilbert space $H$,
and let $F$ be an $\cL(H)$-valued POVM in the sense of
Definition \ref{POVMdef}.
For a bounded function $u\in L^\infty(M)$,
we define the \emph{noise operator}
$$\Delta_F(u) := \int_M u^2\,dF - \left(\int_M u\,dF\right)^2\;.$$
For a pair of bounded functions $u,\,v\in L^\infty(M)$,
set $$U := \int_M u\,dF,\quad V:=\int_M v\,dF\,.$$

\medskip
\noindent
\begin{lemma} \label{lem-unsharp} For every state $\theta \in \cS(H)$, we have the inequality
 \begin{equation}\label{eq-unsharp}
\tr  \left(\Delta_F(u)\theta\right)\,\tr \left(\Delta_F(v)\theta\right)\geq
\frac{1}{4}\,\left|\tr\left( [U,V]\theta\right) \right|^2\;.
 \end{equation}
 \end{lemma}

 \medskip
 \noindent
 \begin{proof}
By the \emph{Naimark dilation theorem}
(see e.g. \cite[Theorem 9.4.6]{PR}), there exists
an isometric embedding of
$H$ into a (possibly infinite-dimensional) Hilbert space $H'$ and
an $\cL(H')$-valued {\it projector valued} measure $P$ such that
for every subset $X \subset \cC$, we have
\begin{equation}\label{Naimdil}
F(X)=\Pi P(X)\Pi^*\in\cL(H)\,,
\end{equation}
where $\Pi : H' \to H$ is the orthogonal projector, so
that $\Pi^*:H\to H'$ is the inclusion.
Here $\cL(H')$ stands for the space of
all bounded Hermitian operators on $H'$, and
an $\cL(H')$-projector valued measure is by definition
a map $P:\cC\to\cL(H')$ satisfying the axioms of Definition
\ref{POVMdef}, and such that the operators $P(X)$, $X \in \cC$, are pair-wise commuting orthogonal projectors.

Define a pairing $$q: \cL(H') \times \cL(H') \to \End(H)\;,$$
$$q(S,T) := \Pi S(1-\Pi^*\Pi)T\Pi^*\;.$$
We claim that for every state $\theta \in \cS(H)$ and all $S,\,T \in \cL(H')$,
we have
\begin{equation}\label{eq-ST-1}
\tr(q(S,S)\theta)\, \tr( q(T,T)\theta) \geq \left|\tr(q(S,T)\theta)\right|^2\;.
\end{equation}
To see this, note that $1-\Pi^*\Pi:H'\to H'$ is the orthogonal projector on the orthogonal complement of $H$, so that
$$\tr(q(S,T)\theta)=
\tr((1-\Pi^*\Pi)T\,\Pi^*\theta\Pi\,S(1-\Pi^*\Pi))\,.$$
Then \eqref{eq-ST-1} follows from
Cauchy-Schwarz inequality
applied to the semi-norm on the space $\End(H')$ of
bounded operators of $H'$ defined by
$(A,B)_\theta:=\tr(A\,\Pi^*\theta\Pi\,B^*)$,
for all $A,\,B\in\End(H')$.

Set now $S= \int u\,dP$ and $T = \int v\,dP$.
Since $S$ and $T$ commute, we have
$[\Pi S\Pi^*,\Pi T\Pi^*] = q(T,S)-q(S,T)$. On the other hand,
by definition \eqref{Naimdil},
we have $\Pi S\Pi^*= U$ and $\Pi T\Pi^*= V$, while using an
approximation
by simple functions, one computes
that $q(S,S) = \Delta_F(u)$ and $q(T,T) = \Delta_F(v)$.
The statement of the lemma then directly follows from \eqref{eq-ST-1}.
\end{proof}

\medskip
\noindent
{\sc Proof of (I):}  Applying Lemma \ref{lem-unsharp} to the
Berezin-Toeplitz POVM $\Gamma_\hbar$ of Definition
\ref{BTdef}, for every state $\theta \in \cS(H_\hbar)$ and observables $u,\,v \in C^\infty(M)$, we get
\begin{equation}\label{eq-surprising}
\tr\left( \Delta_\hbar(u)\theta \right)\,
\tr\left(\Delta_\hbar(v)\theta \right)\\ \geq \frac{1}{4}\,
\left|\tr ([T_\hbar(u),T_\hbar(v)]\theta)\right|^2\;.
\end{equation}
Now by Definition \ref{Bdef} of the Berezin transform
and the expression \eqref{noisedef} for the noise operator,
we know that for all $u\in C^\infty(M)$,
\begin{equation}\label{eq-surp-vsp}
\tr\left ( \Delta_\hbar(u)F_{\hbar,x}\right) = -\hbar\,\cB(c_+(u,u))(x) + \bigo(\hbar^2)\;,
\end{equation}
and for all $u,\,v \in C^\infty(M)$,
\begin{equation}
\begin{split}
-i\,\tr ([T_\hbar(u),T_\hbar(v)]F_{\hbar,x})&= \hbar\,\tr(T_\hbar(\{u,v\})F_{\hbar,x}) +\bigo(\hbar^2)\\
&= \hbar\,B_\hbar(\{u,v\})(x)+\bigo(\hbar^2)\;.
\end{split}
\end{equation}
Thus applying the noise inequality \eqref{eq-surprising}
with $\theta$ being the coherent state $F_{\hbar,x}$,
we get that
\begin{equation} \label{eq-Btr-bound}
\cB_\hbar(c_+(u,u))(x)~\cB_\hbar(c_+(v,v))(x) \geq \frac{1}{4}|\cB_\hbar(\{u,v\})(x)|^2\;,
\end{equation}
so that
the reversibility property (P5) yields
\begin{equation}\label{eqBtrpre}
c_+(u,u)~c_+(v,v) \geq \frac{1}{4}|\{u,v\}|^2\;.
\end{equation}
Thus for all
$\xi,\,\eta\in T_xM$, picking functions
$u,\,v \in C^\infty(M)$ with
$\text{sgrad}\,u(x)=\xi$, $\text{sgrad}\,v(x)=\eta$
and by definition
\eqref{eq-G} of the bilinear form $G$, we get
\begin{equation} \label{eq-Btr-bound-1}
G(\xi,\xi)~G(\eta,\eta) \geq |\omega(\xi,\eta)|^2 \;.
\end{equation}
Now thanks to the non-negativity of the noise
operator, which follows from Lemma \ref{lem-unsharp},
we already know that
$G$ is a semi-positive symmetric bilinear form
by formula \eqref{eq-positivity}.
Inequality \eqref{eq-Btr-bound-1} then shows
that $G$ is in fact positive,
so that it defines a Riemannian metric
on $M$.

Let $K\in\End(TM)$ the
$G$-antisymmetric operator defined by
\begin{equation}
G(\cdot,\cdot)=\om(\cdot,K\cdot)\,.
\end{equation}
Then there exists an orthonormal basis
$\{e_j,f_j\}_{1\leq j\leq\dim M}$ of $TM$ such that
$Ke_j=\alpha_j f_j$ and $Kf_j=-\alpha_je_j$, for $\alpha_j\geq 0$
for all $1\leq j\leq\dim M$. Define an
almost complex structure $J\in\End(TM)$ by the formula
\begin{equation}
Je_j=f_j~~\text{and}~~Jf_j=-e_j\,.
\end{equation}
By definition, this almost complex structure is
compatible with $\om$, and $G$ is $J$-invariant.
Set
\begin{equation}\label{g=G-om}
\rho(\cdot,\cdot):=G(\cdot,\cdot)-\om(\cdot,J\cdot)\,.
\end{equation}
We then need to show that for any $\xi\in TM$, we have
\begin{equation}\label{pos}
\rho(\xi,\xi)\geq 0\,.
\end{equation}
But using \eqref{eq-Btr-bound-1}, we know that
\begin{equation}
G(\xi,\xi)=G(\xi,\xi)^{1/2}\,G(J\xi,J\xi)^{1/2}\geq\om(\xi,J\xi)\,,
\end{equation}
which readily implies \eqref{pos} by definition \eqref{g=G-om}
of $\rho$. \qed

\medskip
\noindent
{\sc Proof of (II):} Recall that
the volume of an $\omega$-compatible metric
is always equal to the symplectic volume $\Vol(M,\om)$.
Then the statement (II) follows from the general form of an
unsharpness metric $G$ given by formula  \eqref{metric}.
\qed

\medskip
\noindent
{\sc Proof of (III):} The construction below is a modification of the one in Example \ref{exam-2}. Instead of dealing with the heat semigroup, which  becomes elusive when
the form $\rho$ is degenerate, we construct an explicit family of Markov kernels
such that the desired quantization is the composition of the almost-K\"{a}hler quantization
associated with $J$ from formula \eqref{metric} with the corresponding Markov operator. \footnote{In the language of quantum measurement theory, the POVM of the quantization
constructed below is a smearing of the Berezin-Toeplitz POVM of
Theorem \ref{BQ} by the explicitly constructed Markov operator.}  Let us pass to precise arguments.

 All the estimates in the proof are meant uniformly in $x_0\in M$.
Let $J\in\End(TM)$ be a compatible almost complex structure on
$(M,\om)$ and let $\rho$ be a non-negative symmetric
bilinear form on $TM$.
Consider the Riemannian metric $g$ over $M$ defined by the formula
\begin{equation}
g(\cdot,\cdot)=\om(\cdot,J\cdot)\,.
\end{equation}
For any $t>0$, we define a smooth endomorphism of the tangent
bundle $TM$ by the formula
\begin{equation}\label{A}
A_{t}:=t\left(-\pi J\rho_g J
+t\id\right)\in\End(TM)\,,
\end{equation}
where $\rho_g\in\End(TM)$ is the non-negative symmetric
endomorphism defined by
\begin{equation}\label{rhog}
g(\rho_g\cdot,\cdot)=\rho\,.
\end{equation}
Then $A_{t}$ is positive symmetric with respect to $g$, for all $t>0$.

Let $\epsilon>0$ be smaller than the injectivity radius of
$(X,g)$. For any $x_0\in X$, consider an
isometric identification
$(T_{x_0}X,g)\simeq(\R^{2d},\langle\cdot,\cdot\rangle)$, where $\langle\cdot,\cdot\rangle$
is the standard Euclidean product of $\R^{2d}$,
and let $Z=(Z_1,\cdots Z_{2d})\in\R^{2d}$ be
the induced normal coordinates on the geodesic ball
$B(x_0,\epsilon)\subset X$
of radius $\epsilon$ centered at $x_0$. We write $dZ$ for the Lebesgue
measure on $\R^{2d}$.
Let $\varphi:[0,+\infty)\rightarrow[0,1]$ be a smooth function
identically equal to $1$ over $[0,\epsilon/2)$ and to $0$ over
$[\epsilon,+\infty)$. We define an operator $K_t^\rho$ acting on
$f\in C^\infty(X,\R)$ by the following formula in normal
coordinates around $x_0\in X$,
\begin{equation}\label{Krhodef}
K_t^\rho f\,(x_0):=\frac{1}{\alpha_{t}(x_0)}\int_{B(x_0,\epsilon)}
\varphi(|Z|)f(Z)\,e^{-\pi\left\langle A_{t}^{-1}Z,Z\right\rangle}
\,dZ\,,
\end{equation}
where $\alpha_{t}(x_0):=\Int_{B(x_0,\epsilon)}\varphi(|Z|)\,
e^{-\pi\left\langle A_{t}^{-1}Z,Z\right\rangle}
\,dZ$ is chosen so that $K_t 1\equiv 1$ for all $t>0$.
Note that $f\geq 0$ implies $K_tf\geq 0$ for all $t>0$.

Fix $x_0\in X$, and consider the isometric identification
$(T_{x_0}X,g)\simeq(\R^{2d},\langle\cdot,\cdot\rangle)$
in which $A_t$ is diagonal, so that
using definition \eqref{A}, we can write
\begin{equation}\label{Adiag}
A_{t,x_0}=\textup{diag}\big(t(\lambda_1+t),\cdots,t(\lambda_{2d}+t)
\big)\,,
\end{equation}
where $\{\lambda_j\geq 0\}_{1\leq j\leq 2d}$ are the eigenvalues
of $-\pi J\rho_gJ$ over $T_{x_0}X$.
Using the multi-index notation
$\alpha=(\alpha_1,\cdots,\alpha_{2d})\in\N^{2d}$,
we will use the following
Taylor expansion of $f$ up to order $4$ as $|Z|\rightarrow 0$,
\begin{equation}\label{Taylorfk}
\begin{split}
f&(Z)  =\sum_{0\leq|\alpha|\leq 3}
\frac{\partial^{|\alpha|} f}{\partial Z^\alpha}(x_0)\frac{Z^\alpha}{\alpha!}
+O(|Z|^4)\|f\|_{C^{4}}\,.\\
\end{split}
\end{equation}
On the other hand, using the change of variables $Z_j\mapsto
Z_j/t^{1/2}(\lambda_j+t)^{1/2}$ for each $1\leq j\leq 2d$
and the exponential decrease of the
Gaussian function, we get
a constant $\delta>0$ for any $\alpha\in\N^{2d}$ such that
the following estimate holds as $t\rightarrow 0$,
\begin{equation}\label{exptrick}
\begin{split}
&\int_{B(x_0,\epsilon)}
\varphi(|Z|)Z^{\alpha}\,e^{-\pi\left\langle A_{t}^{-1}Z,Z\right\rangle}
\,dZ=\int_{\R^{2d}}Z^{\alpha}\,
e^{-\pi\sum_{j=1}^{2d}\left(t^{-1}(\lambda_j+t)^{-1}Z_j^2\right)}
dZ\\
&\quad\quad\quad\quad-\int_{\R^{2d}}(1-\varphi(|Z|))Z^{\alpha}\,
e^{-\pi\sum_{j=1}^{2d}\left(t^{-1}(\lambda_j+t)^{-1}Z_j^2\right)}
\,dZ\\
&=\prod_{j=1}^{2d}t^{1/2}(\lambda_j+t)^{1/2}
\left(t(\lambda_j+t)\right)^{\alpha_j/2}
\int_{\R^{2d}}Z^{\alpha}\,
e^{-\pi|Z|^2}dZ+O(e^{-\delta/t})\,.
\end{split}
\end{equation}
Note that we can then explicitly evaluate the integral in the
last line of \eqref{exptrick} using basic properties of the Gaussian
function, and it vanishes as soon as there is an odd monomial
inside $Z^{\alpha}$.
Then considering the Taylor expansion \eqref{Taylorfk} inside the
right hand side of equation \eqref{Krhodef} and using the estimate
\eqref{exptrick}, we get as $t\rightarrow 0$,
\begin{multline}\label{comput1}
K_t^\rho f\,(x_0)=f(x_0)\\
+\frac{\prod_{k=1}^{2d}t^{1/2}(\lambda_k+t)^{1/2}}{\alpha_{t}(x_0)}\,
\left(\sum_{j=1}^{2d}\frac{t(\lambda_j+t)}{4\pi}
\frac{\partial^2 f}{\partial Z_j^2}(x_0)
+O(t^2)\,\|f\|_{C^{4}}\right)\,.
\end{multline}
On the other hand, it follows from the definition of $\alpha_t$
and the estimate \eqref{exptrick} that as $t\rightarrow 0$, we have
\begin{equation}
\alpha_{t}(x_0)=\prod_{j=1}^{2d}t^{1/2}(\lambda_j+t)^{1/2}(1+
O(e^{-\delta/t}))\,.
\end{equation}
Then we get from equation \eqref{comput1} that as $t\rightarrow 0$,
\begin{equation}\label{comput2}
\begin{split}
K_t^\rho f\,(x_0)&=f(x_0)+t
\sum_{j=1}^{2d}\frac{\lambda_j+t}{4\pi}\,
\frac{\partial^2 f}{\partial Z_j^2}(x_0)
+O(t^2)\,\|f\|_{C^{4}}\\
&=f(x_0)+t
\sum_{j=1}^{2d}\frac{\lambda_j}{4\pi}\,
\frac{\partial^2 f}{\partial Z_j^2}(x_0)
+O(t^2)\,\|f\|_{C^{4}}\,.
\end{split}
\end{equation}
Then writing
$T_\hbar$ for the Berezin-Toeplitz quantization of $(M,\om,J)$,
the quantization $T_\hbar^\rho$ defined for all $f\in C^\infty(X,\R)$ by
\begin{equation}
T_\hbar^\rho(f):=T_\hbar\left(K_\hbar^\rho
f\right)\,,
\end{equation}
has unsharpness metric $G$ given by formula \eqref{metric}: in fact,
for any $u,\,v\in C^\infty(X,\R)$, writing
$\nabla_g u,\,\nabla_g v$ for their gradient with respect to $g$
and in normal coordinates around
$x_0\in X$ as above, we get from the last line of \eqref{comput2}
and following the computations of \eqref{computexam1}
that the unsharpness cocycle $c_+^\rho$ associated with
$T_\hbar^\rho$ satisfies
\begin{equation}
\begin{split}
c_+^\rho(u,v)(x_0)&=-\frac{1}{2}\sum_{j=1}^{2d}\left(
\partial_j u(x_0)\,\partial_j v(x_0)
+\frac{\lambda_j}{\pi}\,\partial_j u(x_0)\,\partial_j v(x_0)\right)\\
&=-\frac{1}{2}\big(g_{x_0}(\nabla_g u,\nabla_g v)-
g_{x_0}(J\rho_gJ\nabla_g u,\nabla_g v)\big)\\
&=-\frac{1}{2}\big(g_{x_0}(\textup{sgrad}\,u,\textup{sgrad}\,v)+
\rho_{x_0}(\textup{sgrad}\,u,\textup{sgrad}\,v)\big)\,.
\end{split}
\end{equation}
This shows that $G=g+\rho$, as required.
\qed

\medskip

\section{Case study: $SU(2)$ - equivariant quantizations} \label{sec-equiv}

\begin{defin}
\label{defin-equiv}{\rm
Two Berezin-Toeplitz quantizations $T_\hbar$ and $T_\hbar'$
with  families of Hilbert spaces $\{H_\hbar\}$ and $\{H'_\hbar\}$, ${\hbar \in \Lambda}$, respectively,
are called {\it equivalent} if
there exists a sequence of unitary operators $U_\hbar: H_\hbar \to H'_{\hbar}$ such that for all $f\in C^\infty(M)$,
\begin{equation}\label{eq-equiv}
\|U_\hbar T_\hbar(f) U_\hbar^{-1} - T'_\hbar(f)\|_{op} = \bigo(\hbar^2)\;.
\end{equation}
}
\end{defin}
Observe that if two quantizations are equivalent, their unsharpness metrics coincide.
In this section we prove a converse statement in the context of $SU(2)$-equivariant quantizations of the two-dimensional sphere (see Section \ref{subsec-class} below for further discussion). We consider the standard K\"{a}hler metric on the two-sphere $S^2$  normalized
so that the total area equals $2\pi$. We denote by $L$ the line bundle dual to the tautological one, and by $H_k$ the $k+1$-dimensional space of holomorphic sections of its $k$-th tensor power $L^k$. One can identify $H_k$ with the space of homogeneous polynomials of two variables, so the group $SU(2)$ acts on $H_k$ via an irreducible unitary representation.
Furthermore, $SU(2)$ acts on the space of Hermitian operators $\cL(H_k)$ by conjugation.
On the other hand the space $C^\infty(S^2)$ carries the natural action of $SU(2)$
by the change of variables. A quantization $Q_\hbar: C^\infty(S^2) \to \cL(H_k)$,
$\hbar\in\Lambda:=\{1/k\}_{k\in\N}$,
is called {\it $SU(2)$-equivariant} if it intertwines the corresponding (real) representations.  For instance, the standard Berezin-Toeplitz quantization $T_\hbar$
sending $f \in C^\infty(S^2)$ to the multiplication by $f$ followed by the orthogonal projection to the space of holomorphic sections is  $SU(2)$-equivariant, and the same holds true for its images $T_\hbar^{(t)}$ under diffusion as defined in Example \ref{exam-2}.
Note that the quantizations $T_\hbar^{(t)}$ are pair-wise non-equivalent for different values of $t$ as the corresponding unsharpness metrics are different.

\begin{thm} \label{thm-equiv} Every  $SU(2)$-equivariant
quantization of $S^2$ is equivalent to $T_\hbar^{(t)}$ for some $t \geq 0$.
\end{thm}
\begin{proof}
{\sc Step 1 (Applying Schur lemma):} Given any $SU(2)$-equivariant quantization $Q_\hbar$, pass to its complexification (denoted by the same letter)
$$Q_\hbar: C^\infty(S^2,\C) \to \cL(H_k) \otimes \C = H_k^* \otimes H_k\;.$$
On the one hand, $C^\infty(S^2,\C)$ splits into the direct sum of irreducible summands $V_j$, $j=0,1,\dots$ corresponding to the eigenspaces of the Laplace-Beltrami operator associated to the K\"{a}hler metric with the eigenvalue $2j(j+1)$, with each $V_j$ isomorphic to $H_{2j}$ as an $SU(2)$-representation. On the other hand
$$H_k^* \otimes H_k = H_{2k} \oplus H_{2k-2} \oplus \cdots \oplus H_0\;.$$
By the Schur Lemma, when $\hbar = 1/k$, we have that $Q_\hbar(V_j) \subset H_{2j}$
with respect to this decomposition, and furthermore there exists a constant $\alpha_{\hbar, j}\in\C$
such that, up to conjugation, we have
\begin{equation} \label{eq-QT}
Q_\hbar = (1+\alpha_{\hbar, j})T_\hbar \;\; \text{on}\;\; V_j\;.
\end{equation}

{\sc Step 2 (Legendre polynomials):} In what follows we introduce another parameter, $n \in \N$. We call a sequence $\{b_{\hbar,n}\}_{n\in\N}$
of the class $\bigo_N(\hbar^m)$ with $m,N \in \N$
if for some $c>0$ we have $|b_{\hbar,n}| \leq c\,\hbar^m (n+1)^N$ for all $n$. In the sequel, the dependence on $\hbar$ of such
sequences will be made implicit.
Denote by $P_n(z)$ the $n$-th Legendre polynomial considered as a function on the unit sphere
$S^2 = \{x^2+ y^2 +z^2 =1\}$ lying in $V_n$.  We write $\nabla$ for the gradient with respect to the standard metric on $S^2$ normalized so that the total area equals $2\pi$. The standard formulas for the Legendre polynomials (see e.g. formulas (43) and (44) in \cite{W}) readily
yield\;,
\begin{equation}\label{Leg-1}
P_1P_n = q_n P_{n+1} + r_nP_{n-1},\;\; q_n=\frac{n+1}{2n+1},\;\; r_n= 1-q_n\;,
\end{equation}
and
\begin{equation}\label{Leg-2}
(\nabla P_1, \nabla P_n) = s_n ( -P_{n+1} + P_{n-1}),\;\; s_n=\frac{2n(n+1)}{2n+1}\;.
\end{equation}
We shall use that there exists $c>0$ such that
\begin{equation}
\label{eq-Legendre-deriv}
\forall r \in \N \;\; \exists R\in \N\;: \|P_n\|_{C^r} \leq c(n+1)^R\;.
\end{equation}
This (with $R = r$) follows immediately from the general result about the growth
of $C^r$-norms of the Laplace-Beltrami eigenfunctions on Riemannian manifolds,
see \cite[Corollary 1.1]{BinXu}.
Using the fact that $\max_{x\in S^2} P_n=1$ by
\cite[Chapter 7, Theorem 17(i)]{Ka}, the norm correspondence property
(P1), which holds uniformly in $C^N$-norm for some $N \in \N$,
together with formula \eqref{eq-Legendre-deriv} implies
$$\|Q_\hbar(P_n)\|_{op} = 1-\bigo_N(\hbar),\;\; \|T_\hbar(P_n)\|_{op} = 1-\bigo_N(\hbar)\;.$$
Since $Q_\hbar(P_n) = (1+ \alpha_{\hbar,n})T_\hbar(P_n)$
by \eqref{eq-QT}, it follows that
\begin{equation}
\label{eq-alpha-O} \alpha_{\hbar,n} = \bigo_N(\hbar)\;.
\end{equation}
In the course of the proof, we shall increase the value of $N$ according to
our needs.

\medskip\noindent
{\sc Step 3 (Main calculation) :} Since $Q_\hbar$ is $SU(2)$-equivariant,
the corresponding unsharpness metric equals $\mu$ times the standard one,
for some constant $\mu \geq 1$. Thus the quasi-multiplicativity property
(P3), which holds uniformly in $C^N$-norm for some $N\in\N$,
together with formula \eqref{eq-Legendre-deriv} yields
\begin{equation}
\label{star-enhanced-3}
Q_\hbar(P_1)Q_\hbar(P_n) = Q_\hbar\left(P_1P_n- \frac{\mu}{2}\hbar(\nabla P_1,\nabla P_n) + \bigo_N(\hbar^2)\right)\;.
\end{equation}
At the same time
\begin{equation}
\label{star-enhanced-4}
T_\hbar(P_1)T_\hbar(P_n) = T_\hbar\left(P_1P_n- \frac{1}{2}\hbar(\nabla P_1,\nabla P_n) + \bigo_N(\hbar^2)\right)\;,
\end{equation}
mind that here $\mu$ is replaced by $1$.
By \eqref{eq-QT} we have
\begin{equation} \label{eq-QT-1}
Q_\hbar(P_i) = (1+\alpha_{\hbar, i})T_\hbar(P_i)\;.
\end{equation}
Identities \eqref{star-enhanced-3} and \eqref{star-enhanced-4} combined with
\eqref{Leg-1},\eqref{Leg-2} and \eqref{eq-QT-1} enable us to express $T_\hbar(P_1)T_\hbar(P_n)$ as a linear
combination of $T_\hbar(P_{n+1})$ and $T_\hbar(P_{n-1})$ in two different ways. The calculation
is straightforward, and we obtain the result:
\begin{equation}\label{eq-vsp-PP-1}
A_nT_\hbar(P_{n+1})+ B_n T_\hbar(P_{n-1}) = A'_nT_\hbar(P_{n+1})+ B'_n T_\hbar(P_{n-1})+ \bigo_N(\hbar^2)\;,
\end{equation}
for some $A_n,\,A_n'\in\C$ and $B_n,\,B_n'\in\C$, $n\in\N$,
where
$$B_n = (1+\alpha_{\hbar, 1})^{-1} (1+\alpha_{\hbar, n})^{-1}(1+\alpha_{\hbar,n-1})(r_n -\hbar\mu s_n/2)\;,$$
$$B'_n= r_n - \hbar s_n /2\;.$$
Projecting equation \eqref{eq-vsp-PP-1} to the space $H_{2n-2}$ (which contains $T_\hbar(V_{n-1})$) and using that the operator norm of
$T_\hbar(P_{n-1})$ is bounded away from zero (see Step 2), we get that
$$B_n - B'_n = \bigo_N(\hbar^2)\;.$$
By using \eqref{eq-alpha-O} and explicit expressions for $q_n,r_n,s_n$  we get
\begin{equation}\label{eq-recur-2}
\alpha_{\hbar,n-1} - \alpha_{\hbar, n} - \alpha_{\hbar, 1} = (n+1)(\mu-1)\hbar + \bigo_N(\hbar^2)\;.
\end{equation}
Substituting  $n=1$ into \eqref{eq-recur-2} we get that
$$\alpha_{\hbar,1}= -(\mu-1)\hbar + \bigo_N(\hbar^2)\;.$$
Now we get a recursive formula
$$\alpha_{\hbar,n} =  \alpha_{\hbar,n-1} -n(\mu-1)\hbar + \bigo_N(\hbar^2)\;.$$
Noticing that $(n+1)\,\bigo_N(\hbar^2)= \bigo_{N+1}(\hbar^2)$ and redefining $N \mapsto N+1$ we
conclude that
\begin{equation}\label{eq-alpha-complete}
\alpha_{\hbar, n} = - \frac{n(n+1)}{2}(\mu-1)\hbar + \bigo_N(\hbar^2)\;.
\end{equation}

{\sc Step 4 (Finale) :} Recall that $2n(n+1)$ is the eigenvalue of the Laplacian corresponding
to the eigenspace $V_n$. Let $V = \oplus_{n=0}^\infty V_n$ be the space of all finite linear combinations
of spherical harmonics. By norm correspondence (P1) and
formula \eqref{eq-alpha-complete},
for every $\phi_n \in V_n$ we have
$$Q_\hbar(\phi_n) = \left(1- \frac{n(n+1)}{2}(\mu-1)\hbar \right) T_\hbar (\phi_n) + \bigo_N(\hbar^2)\,\|\phi_n\|_{C^N}
$$ $$ =T_\hbar(e^{-t\hbar\Delta} \phi_n) + \bigo_N(\hbar^2)\,\|\phi_n\|_{C^N} = T^{(t)}_\hbar(\phi_n)+ \bigo_N(\hbar^2)\,\|\phi_n\|_{C^N}$$
with $t = (\mu-1)/4$ in Example \ref{exam-2}.

Take now any $f \in C^\infty(S^2)$, and decompose it by spherical harmonics: $f = \sum_n \phi_n$.
Since $f$ is smooth, the $C^N$-norms
$\|\phi_n\|_{C^N}$ decay faster than any power of $n$
as $n\to+\infty$, so that
$$\|Q_\hbar(f)- T_\hbar^{(t)}(f)\|_{op} \leq c\sum_{n\in\N} n^N \|\phi_n\|_{C^N}\hbar^2 \leq c'\hbar^2\;.$$
This shows that the quantizations $Q_\hbar$ and $T_\hbar^{(t)}$ are equivalent.
\end{proof}

\section{The unsharpness cocycle is of order $(1,1)$}
\label{sec-order-proof}

In this Section we prove Theorem \ref{thm-order}

\begin{proof}
For every $d\in\N$, we use the standard multi-index notation
$\alpha=(\alpha_1,\cdots,\alpha_d)\in\N^d$, where
for any sequence of symbols $x_1,\cdots, x_d$, we write
$x_\alpha:=x_{\alpha_1}\cdots x_{\alpha_d}$, so that in
particular $\alpha!:=\alpha_1!\cdots\alpha_d!\in\N$, and
write $|\alpha|:=\alpha_1+\cdots+\alpha_d$.

Note that by \eqref{eq-a},
to show that $c_+$ is a bi-differential operator of bi-degree
$(1,1)$, we need to show that
the differential operator $a$ contains only terms
of order $1$ and $2$. Note that $T_\hbar(1)=\id$ implies
$c_+(1,1)=0$, so that $a$ cannot contain terms of order $0$.
Let us show that $a$ cannot be of order $k>2$.

Assume by contradiction that $a$ is of order $k>2$. Let
$x_0\in X$ be the center of local coordinates
$(Z_1,\cdots,Z_{2n})\in U\subset\R^{2n}$ be such that
for all $f\in C^\infty(X,\R)$,
\begin{equation}\label{Ldef}
af(x_0)=\sum_{1\leq|\alpha|\leq k}a_\alpha\frac{\partial^{|\alpha|}f}
{\partial Z^\alpha}(x_0)\,,
\end{equation}
where the sequence $\{a_\alpha\in\R\}_{1\leq|\alpha|\leq k}$ is such that
$a_\beta\neq 0$ for some $\beta=(\beta_1,\cdots,\beta_{2n})\in\N^{2n}$
of length $|\beta|=k$.
Fix $1\leq j\leq 2n$ such that
$\beta_j\neq 0$, and writing
$\hat\beta=(\beta_1,\cdots,\beta_j-1,\cdots,
\beta_{2n})\in\N^{2n-1}$,
take $f\in C^{\infty}(X,\R)$ satisfying
\begin{equation}\label{fcoord}
f(Z)=\frac{c}{\hat\beta!}Z^{\hat\beta}+Z_j\,,
\end{equation}
for $Z\in U\subset\R^{2n}$ in the coordinates around
$x_0\in X$ considered above and for some $c\in\R$ to be fixed later.
Then this function $f$ and all its derivatives
vanish at $x_0\in X$, except for
\begin{equation}\label{fcoord-1}
\frac{\partial^{|\hat\beta|}f}
{\partial Z^{\hat\beta}}(x_0)=c\quad\quad\text{and}
\quad\quad\frac{\partial f}{\partial Z_j}(x_0)=1\,.
\end{equation}
Then by equations \eqref{fcoord} and \eqref{Ldef}, considering
the multi-index $\gamma\in\N^{2n}$
of length $|\gamma|=2$ such that $\gamma_j=2$,
one gets that for any $f\in C^\infty(X,\R)$ satisfying
\eqref{fcoord-1},
\begin{equation}
\begin{split}
c_+(f,f)(x_0)&=2a_\beta\,\frac{\partial^{|\hat\beta|}f}
{\partial Z^{\hat\beta}}(x_0)\,\frac{\partial f}{\partial Z_j}(x_0)
+2a_\gamma\frac{\partial f}{\partial Z_j}(x_0)\frac{\partial f}{\partial Z_j}(x_0)\\
&=2a_\beta c+2a_\gamma\,.
\end{split}
\end{equation}
Thus if $f\in C^\infty(X,\R)$ satisfies \eqref{fcoord-1}
for $c\in\R$ such that
$\text{sign}(a_\beta)c > -a_\gamma/|a_\beta|$, we get that $c_+(f,f)(x_0)> 0$.
This contradicts the fact that $c_+(f,f)\leq 0$ for all
$f\in C^\infty(X,\R)$, which holds for every  Berezin-Toeplitz quantization by the semi-positivity property
\eqref{eq-positivity}.
\end{proof}

\section{Discussion and questions}\label{sec-disc}

\subsection{Historical remarks on unsharpness}\label{subsec-ger} The unsharpness cocycle appeared in earlier literature which, to the best of our knowledge, focussed on its elimination, of course, by the price of losing the positivity of a quantization. Let us elaborate on this point.
Assume $(M,\omega)$ is a quantizable symplectic manifold equipped with
a  compatible almost-K\"{a}hler structure. Consider the induced
Berezin-Toeplitz quantization of Theorem \ref{BQ}.
Using the notations of Example \ref{exam-2}, define for any $f\in C^\infty(M)$ and
$\hbar\in\Lambda=\{1/k\}_{k\in\N}$,
\begin{equation}\label{metaKS}
Q_\hbar(f) := T_\hbar\left(f+\frac{\hbar}{4} \Delta f\right)\;.
\end{equation}
This gives rise to a collection of maps
$Q_\hbar:C^\infty(M)\to\cL(H_\hbar)$ parametrized by
$\hbar\in\Lambda$ and satisfying the axioms (P1)-(P4) of Definition
\ref{BTdef}, but which does not preserve positivity,
so that they do not come from a POVM construction via formula
\eqref{quantmap}. Then following the computation \eqref{computexam1}
in Example \ref{exam-1}, we see that the associated
unsharpness cocycle $c_+^Q$, defined from the quasi-multiplicativity property (P3) as in the beginning of the section,
satisfies
\begin{equation}\label{c+=0}
c_+^Q(f,g)=0\,,
\end{equation}
for all $f,\,g\in C^\infty(M)$. As noted for instance by Charles in
\cite[\S\,1.4]{Cha} \footnote{\cite{Cha} uses the holomorphic Laplacian, which is half the Laplace-Beltrami operator.}, the quantization \eqref{metaKS} is,  up to twisting with
a line bundle, the \emph{metaplectic Kostant-Souriau quantization}, which
possesses remarkable sub-principal properties, a fact which
is explained conceptually by the vanishing unsharpness
property \eqref{c+=0}.

In the flat case $M=\C$ with the standard symplectic form,
Gerstenhaber considers in \cite{G} deformation quantizations
parametrized by $\lambda\geq 0$
which, up to the second order in $\hbar$, correspond to the quantization maps parametrized
by $\hbar>0$ defined for any smooth function $f: \C \to \R$ of polynomial growth by
\begin{equation}\label{Gerquant}
Q_\hbar^{(\lambda)}(f) := T_\hbar\left(f+ \frac{1-\lambda}{2}\hbar
\Delta f\right)\;.
\end{equation}
Here $T_\hbar$ is the standard Toeplitz quantization of $\C$,
sending $f$ to the multiplication by $f$
followed by the orthogonal $L_2$-projection
on the space of
holomorphic functions which are
square integrable with respect to a Gaussian
measure. Gerstenhaber formulates a \emph{least uncertainty principle}
for deformation quantization, which implies in particular
that unsharpness vanishes on the classical harmonic oscillator.
He then shows
that the quantization \eqref{Gerquant} satisfies this least
uncertainty principle if $\lambda=1/2$, which corresponds
to the flat version of the quantization \eqref{metaKS}.

Note that in the flat case $M=\C$, the classical harmonic oscillator is
a sum of squares of the coordinate functions. On the other hand,
the quasi-multiplicativity property (P3) implies that
for all $f\in C^\infty(M)$ as $\hbar\to 0$,
$$T_\hbar(f)^2-T_\hbar(f^2)=\hbar T_\hbar(c_+(f,f))+\bigo(\hbar^2)\,.
$$
We then see that unsharpness measures in particular
the deviation of the quantum harmonic oscillator, defined
as a sum of squares of the quantum coordinate operators, from
the quantization of the classical harmonic oscillator.
This explains in particular the standard justification of the
metaplectic correction, as giving the "correct" quantum harmonic
oscillator on flat space.

\subsection{Least unsharpness surfaces and pseudo-holomorphic curves}
Let $G$ be the unsharpness metric associated to
a Berezin-Toeplitz quantization of a closed symplectic manifold $(M,\omega)$
(see Section \ref{sec-main}). {\it A least unsharpness  surface} $\Sigma \subset M$ is
a two-dimensional oriented submanifold with $\text{Area}_G (\Sigma) = \int_\Sigma \omega$.
Repeating the the proof of Theorem \ref{thm-main-new-1} we see that for such surfaces, the restriction of the Riemannian area form coincides with the restriction of the symplectic form. If $G$ has the minimal possible total unsharpness  and hence by Theorem  \ref{thm-main-new-1}\,(II) comes from some
compatible almost-complex structure $J$ on $M$, the least unsharpness  surfaces in $M$
are $J$-holomorphic curves (cf. \cite{S}).  For instance, for the complex projective plane $M= \C P^2$, Gromov's theory of pseudo-holomorphic curves predicts that for every compatible $J$,  through every two distinct points $A,B \in M$ passes unique such curve $\Sigma$ in the homology class of $[\C P^1]$.

It is enticing to interpret $\Sigma$ as a worldsheet of the topological string theory describing a path joining constant loops $A$ and $B$. Note that the metric $G$ on our ``space-time" $M$ is canonically associated to a Berezin-Toeplitz quantization of $M$, and the ``total unsharpness " $\text{Area}_G (\Sigma)$ of a worldsheet $\Sigma$ is nothing else but the Nambu-Goto action up to a multiplicative constant. If the total unsharpness  of $(M,G)$ is minimal possible, i.e., coincides with the symplectic volume of $M$, the least unsharpness  surfaces are $J$-holomorphic curves for a compatible almost complex structure $J$ defining $G$, and hence represent ``worldsheet instantons". Does there exist an interpretation of this
picture in physical terms?

\subsection{On classification of Berezin-Toeplitz quantizations}\label{subsec-class}
We conclude the paper with a discussion on classification of Berezin-Toeplitz quantizations
up to equivalence in the sense of Definition \ref{defin-equiv}. In Section \ref{sec-equiv} we classified $SU(2)$-equivariant quantizations of the two-dimensional sphere. It would be interesting to extend this to equivariant quantizations for more general co-adjoint orbits equipped with the canonical symplectic structure. In the general (not necessarily equivariant) case, the
problem is widely open.

In fact,  establishing (non)-equivalence of quantizations is a non-trivial problem even  for the K\"{a}hler quantizations of Theorem \ref{BQ}, where the holomorphic line bundles defining the quantization of $(M,\om)$ could be non-isomorphic. For instance, their Chern classes could differ by torsion even though the associated spaces of holomorphic sections have same dimension. Are the corresponding quantizations equivalent?

Another interesting example is as follows.  According to Remark \ref{rem-rho},
there exist metrics $G$ on $M$ admitting different decompositions of the form \eqref{metric}.
Each such decomposition determines a Berezin-Toeplitz quantization given by almost-K\"{a}hler
quantization followed by diffusion, as explained in the proof of Theorem \ref{thm-main-new-1}\,(III). Are the quantizations corresponding to different decompositions
of the same metric equivalent?

Let us address the question about invariants of quantizations with respect to equivalence.
In addition to the unsharpness metric,  there is another invariant coming from the trace correspondence, see item (P4) in Definition \ref{BTdef}. Recall that the latter states that
$$  \tr(T_\hbar(f)) = (2\pi\hbar)^{-d}\int_M f\,R_\hbar\, d\mu\;,$$
where $\dim M = 2d$, the function $R_\hbar\in C^\infty(M)$ satisfies $R_\hbar = 1+ \bigo(\hbar)$, and $d\mu = \frac{\omega^d}{d!}$ is the symplectic volume on $M$. Roughly speaking, since the trace is invariant under conjugation, the convergence rate of the sequence of differential forms $R_\hbar d\mu$ to the symplectic volume $d\mu$ as $\hbar \to 0$ does not change up to $\bigo(\hbar^2)$ under equivalence. For the sake of simplicity, let us, until the end of the paper,
enhance axiom (P4) by assuming that there exists a function $r \in C^\infty(M)$
such that $$R_\hbar = 1 + \hbar\,r + \bigo(\hbar^2)\;.$$ We shall refer to the form $rd\mu$ as {\it the Rawnsley form}. Thus, {\it equivalent quantizations possess the same Rawnsley form}. Put
$$\langle r \rangle := \text{Vol}(M)^{-1} \int_M r d\mu\;.$$

\begin{rem}{\rm Substituting $f=1$ into (P4), we get that
$\langle r \rangle $ appears in the dimension formula
$$\dim H_\hbar = \text{Vol}(M)(2\pi \hbar)^{-d}\left(1 +
\hbar\,\langle r \rangle + \bigo(\hbar^2)\right)\;.$$
Let us mention that for K\"{a}hler quantizations, an alternative asymptotic expression for the
dimension of $H_\hbar$ is given by the Hirzebruch-Riemann-Roch theorem. Comparing
coefficients at $\hbar$ one gets a simple topological interpretation of $\langle r \rangle$:
$$\langle r \rangle = 2\pi\,\frac{\left\langle [\omega]^{d-1}\cup c_1(TM), [M]\right\rangle}{2(d-1)!  \text{Vol}(M)} \;,$$
where $c_1(TM)$ stands for the first Chern class of $M$. }
\end{rem}

\begin{exam}
\label{exam-rawn}
{\rm Let $v$ be a vector field on the manifold $M$ generating a flow $\phi_t$.
Given a Berezin-Toeplitz quantization $T_\hbar$ on $M$, define a new quantization
by setting $T^{(v)}_{\hbar}(f) : = T_\hbar(f\circ \phi_{-\hbar})$. A direct calculation
based on the expansion
$T^{(v)}_{\hbar}(f) = T_\hbar\left(f -\hbar\,df(v)\right) + \bigo(\hbar^2)$
shows that this is a Berezin-Toeplitz quantization whose unsharpness metric coincides with
the one of $T_\hbar$, and whose Rawnsley form is given by
$(r + \text{div} (v))d\mu$,
where $\text{div} (v)$ stands for the divergence of $v$ with respect to the symplectic volume.
In particular, it follows that by choosing an appropriate vector field $v$, one can always
achieve the Rawnsley form being equal to $\langle r \rangle$.
}
\end{exam}


\begin{question}\label{question-3} {\rm Consider a pair of quantizations with
the Hilbert spaces of the same dimension. Suppose that their unsharpness metrics
and the Rawnsley forms coincide. Are these quantizations equivalent?}
\end{question}

The answer in the general (not necessarily equivariant) case is at the moment unclear.

\medskip
\noindent
{\bf Acknowledgement.} We thank Pavel Etingof for useful references. We are grateful to Jordan Payette for pointing out a number of inaccuracies and helpful suggestions on the presentation.

\Addresses

\end{document}